%% file: issac-report.tex
\documentclass{article}
\usepackage{balance}
\usepackage{doi}
\usepackage{url}

\def\tit{Generic reductions for in-place polynomial multiplication}

\usepackage{multirow}
\usepackage{etoolbox}

\usepackage{enumitem}
\usepackage[capitalise,noabbrev,nameinlink]{cleveref}
\Crefname{enumi}{Step}{Steps}
\usepackage{todonotes}

\usepackage{hyperref}
\usepackage{amsmath,amsfonts,amssymb,amsthm,bm}
\usepackage{enumitem}

\usepackage{boxedminipage} 

\usepackage{algorithm}
\usepackage[noend]{algpseudocode}
\algnewcommand\Input{\item[\textbf{Input:}]}
\algnewcommand\Output{\item[\textbf{Output:}]}
\algnewcommand\BaseAlg{\item[\textbf{Required alg.:}]}

\usepackage{tikz}
\usepackage{xspace}

\usepackage{appendix}

\usepackage[english]{babel}
\usepackage[utf8]{inputenc}

\newtheorem{thm}{Theorem}[section]   
\newtheorem{lem}[thm]{Lemma}
\newtheorem{cor}[thm]{Corollary}

\newtheorem{definition}[thm]{Definition}

\newcommand{\K}{\ensuremath{\mathbb{K}}} 
\newcommand{\M}{{\sf M}}

\newcommand{\bigO}[1]{\ensuremath{O(#1)}} 

\newcommand\algo[1]{\ensuremath{\mathtt{#1}}\xspace}
\newcommand\fct[1]{\ensuremath{\mathsf{#1}}\xspace}
\newcommand\reduction[2]{\algo{#1\_from\_#2}}

\newcommand\FP{\fct{FP}}
\newcommand\SP{\fct{SP}}
\newcommand\MP{\fct{MP}}
\newcommand\FPhA{\FP^+}
\newcommand\SPhi{\fct{SP_{hi}}}
\newcommand\SPlo{\fct{SP_{lo}}}
\newcommand\FPlo{\fct{\FPhA_{lo}}}
\newcommand\FPhi{\fct{\FPhA_{hi}}}

\newcommand\oFPunbalhA{\fct{oFP^+_u}}

\newcommand\oFP{\algo{oFP}}
\newcommand\iFPhi{\algo{iFP^+_{hi}}}
\newcommand\oMP{\algo{oMP}}
\newcommand\iMP{\algo{iMP}}
\newcommand\oSP{\algo{oSP}}
\newcommand\oSPlo{\algo{oSP_{lo}}}
\newcommand\oSPhi{\algo{oSP_{hi}}}
\newcommand\iSPlo{\algo{iSP_{lo}}}

\newcommand\mat[1]{\ensuremath{\mathfrak{M}_{#1}}\xspace}

\newcommand\TISP{\ensuremath{\mathsf{TISP}}\xspace}

\newcommand\redto{\ensuremath{\le_\TISP}\xspace}
\newcommand\eqto{\ensuremath{\equiv_\TISP}\xspace}

\def\quad{\hskip1em\relax}
\def\qquad{\hskip2em\relax}

\newcommand\out[1][R]{\texttt{#1}}
\DeclareMathOperator{\quo}{quo}
\DeclareMathOperator{\rev}{rev}

\newcommand\floor[1]{\left\lfloor #1\right\rfloor}
\newcommand\ceil[1]{\left\lceil #1\right\rceil}

\usepackage{color}

\newcommand{\algoname}[1]{{\normalfont\textsc{#1}}}

\usepackage{graphicx}
\graphicspath{{./}}
\makeatletter
\newcommand\svg[2][]{\ifx#1\@empty\else\def\svgwidth{#1}\fi\input{#2-svg.tex}}
\makeatother

\newcommand*{\keywords}[1]{\textbf{\textit{Keywords---}} #1}

\def\grantsponsor#1#2#3{#2}
\newcommand\grantnum[3][]{#3%
  \def\@tempa{#1}\ifx\@tempa\@empty\else\space(\url{#1})\fi}

\begin{document}
\title{\tit}

\author{Pascal Giorgi\thanks{  LIRMM, Université de Montpellier,
    CNRS,Montpellier, France. \url{(firstname.lastname@lirmm.fr)}},
    Bruno Grenet\footnotemark[1] and 
    Daniel S.\ Roche\thanks{  United States Naval Academy, Annapolis, Maryland,
      USA. \url{roche@usna.edu}}
  }

\maketitle

\begin{abstract}
  The polynomial multiplication problem has attracted considerable attention since the
  early days of computer algebra, and several algorithms have been designed to
  achieve the best possible time complexity. More recently, efforts have
  been made to improve the space complexity, developing modified
  versions of a few specific algorithms to use no extra space while
  keeping the same asymptotic running time.

  In this work, we broaden the scope in two regards. First, we ask
  whether an arbitrary multiplication
  algorithm can be performed in-place generically. Second, we consider two important
  variants which produce only part of the result (and hence have less
  space to work with), the so-called middle and short products, and ask
  whether these operations can also be performed in-place.

  To answer both questions in (mostly) the affirmative, we provide a
  series of reductions starting with any linear-space multiplication
  algorithm.  For full
  and short product algorithms these reductions yield in-place versions with the
  same asymptotic time complexity as the out-of-place version.
  For the middle product, the reduction incurs an extra logarithmic
  factor in the time complexity only when the algorithm is quasi-linear.
\end{abstract}


\section{Introduction}

\subsection{Polynomial multiplication}

Polynomial multiplication is a fundamental problem in mathematical
algorithms. It forms the basis (and key bottleneck) for
other fundamental problems such as division with remainder, GCD computation,
evaluation/interpolation, resultants, factorization, and structured linear algebra
(see, e.g., \cite[\S 8--15]{Gathen:2013} and
\cite[\S 2--7,10,12]{BCG+17}).

As such, significant effort has gone to
improving the time to multiply two size-$n$ polynomials, most
notably Karatsuba's algorithm \cite{karatsuba}, Toom-Cook multiplication
\cite{Cook66}, and Sch\"onhage-Strassen \cite{SS71}; more recent results
have improved the complexity further but have not yet seen wide
adoption in practice \cite{CK91,HHL17}.

\subsection{Space complexity}

After minimizing the runtime, an important question both in theory and
in practice is how much extra \emph{space} these algorithms require.
While the classical algorithm
can be made to use only a constant number of temporary
values, all the faster algorithms mentioned above require $\bigO{n}$
space to multiply two size-$n$ polynomials. In fact, proven time-space
trade-offs in the algebraic circuit and branching program models indicate
that space at least polynomial in $n$ is required for any sub-quadratic
multiplication algorithm \cite{SS79,Abr86}.

But in a model where the output space admits both random writes
\emph{and} reads, these time-space lower bounds can be broken.
\cite{roche:2009} developed
a variant of Karatsuba's algorithm using only
$\bigO{\log n}$ space. Later, an FFT-based
multiplication algorithm using $\bigO{n\log n}$ time and constant space
was developed for the case that the coefficient ring contains a suitable
root of unity \cite{hr10}. Space-saving versions of Karatsuba's
algorithm can also be found in
\cite{Thome:2002,Brent:2010:MCA,SF12,Che16}.

\subsection{Short and middle products}

Besides the usual \emph{full product} computation,
two other variants have also been extensively
studied: the \emph{short product} which truncates the output to the
first $n$ terms, and the \emph{middle product}
which truncates the result on both
ends. These variants are important especially for
power series, and specific variants of Karatsuba's
algorithm and others have been developed, usually gaining
a constant factor compared to a full product followed by a
truncation \cite{HQZ00,Mulders:2000,Hanrot2004,HanrotZimmerman:2004}.

\cite{Bostan:2003} shows that the middle
product can be viewed essentially as the reverse of a full product and
in the same space. However, in our model which uses the space of the
output as temporary working space, this reversal implies that the inputs
must also be destroyed for an in-place middle product.
In some sense
it would not be surprising if middle and short products were more
difficult in our setting, as the truncated size of the output
essentially limits the working space of the algorithm.

\subsection{Our work}

In this paper, we develop reductions which can transform any
multiplication algorithm which uses $\bigO{n}$ extra space into full,
short, and middle product algorithms which use only $\bigO{1}$ extra
space. The time complexity for full and short product is the same as
that of the original, while that for middle product incurs an additional
$\log n$ factor.

This improves the \(\bigO{\log n}\) space of the most
space-constrained Karatsuba algorithm \cite{roche:2009}, and implies for the
first time: in-place versions of Toom-Cook multiplication; in-place
FFT-based multiplication even when the ring does not contain a root of
unity; in-place subquadratic short product algorithms; and in-place middle product
algorithms which do not overwrite their inputs.

We begin by carefully stating our space complexity model and then
defining the multiplications problems in \cref{sec:models,sec:polymul}.
A few easier but
important reductions and equivalences are presented next in
\cref{sec:tisp-preserving}, followed by the critical reductions in
\cref{sec:out2in} which prove our main results.

\section{Complexity model} \label{sec:models}
We use the model
of an \emph{algebraic-RAM} that is equipped with two kinds of
registers: the standard registers store integers as in the classical Word-RAM
model, whereas the \emph{algebraic registers} store elements from the base field $\K$
of coefficients.
As in Word-RAM, we assume that the standard
registers can store integers of size $\bigO{\log n}$ where $n$ is the number of
coefficients in the inputs.

Word-RAM machines are a classical model in computational complexity, in
particular for \emph{fine-grained complexity} that classifies the difficulty of
polynomial-time problems \cite{Va18}.  We use it in order to distinguish between
the space needed to store indices (that is thus hidden in the standard
registers) from the space needed to store elements from the base field.

\paragraph{Time complexity}
As mentioned, we use the number of arithmetic operations as the time complexity
measure since the cost of the operations on indices is negligible with respect
to arithmetic operations. Formally, we assume that any ring operation on the
algebraic registers has cost $1$.

\paragraph{Space complexity}
We divide the registers into three categories: the input space
is made of the (algebraic) registers that store the inputs, the output space is
made of the (algebraic) registers where the output must be written, and the work
space is made of (algebraic and non-algebraic) registers that are used as extra
space during the computation. The space complexity is then the maximum number of
work registers used simultaneously during the computation. An algorithm is said
to be ``in-place'' if its space complexity is $\bigO 1$, and ``out-of-place''
otherwise.

One can then distinguish different models depending on the read/write
permissions on
the input and output registers:
\begin{enumerate}
\item Input space is read-only, output space write-only;
\item Input space is read-only, output space is read/write;
\item Input and output spaces are both read/write. 
\end{enumerate}
The first model is the classical one from complexity theory~\cite{ArBa09}.
Despite its theoretical interest,
it does not reflect low-level computation where output is typically
in some DRAM or Flash memory on which reading is no more costly
than writing.  Furthermore, polynomial
multiplication here has a quadratic lower bound for time times space
\cite{Abr86}, limiting the possibility for meaningful improvements.

The second model has been used in the context of in-place polynomial
multiplication~\cite{roche:2009,hr10}. This is a very reasonable model
since it matches the paradigm of parallel computing with shared memory.
This is the model in which we develop our algorithms.

The third model has been used to
provide a generic approach for preserving memory
designing algorithms via the transposition principle~\cite{Bostan:2003}: Given an
algorithm for a linear map
with time complexity $t(n)$ and space complexity $s(n)$, the
transposition principle yields an algorithm for the transposed linear
map which has the same space complexity and time complexity
$\bigO{t(n)}$~\cite[Propositions 1 and 2]{Bostan:2003}.
However, the inputs are destroyed during the computation, which is
problematic particularly for recursive algorithms that re-use their
operands; we will not use this too-permissive model.

\paragraph{Notation} 
The output space in our algorithms is denoted by $\out$ and registers
are indexed from $0$ to $n-1$. We write $\out_{[k..\ell[}$ to denote the registers of indices $k$ to $\ell-1$.

\section{Polynomial multiplications} \label{sec:polymul}

Define the \emph{size} of a univariate polynomial as the number of
coefficients in its (dense) representation; a polynomial of size $n$ has
degree at most $n-1$. Importantly, we allow zero padding: a size-$n$
polynomial could have degree strictly less than $n-1$; the size
indicates only how it is represented.

Let $f=\sum_{i=0}^{n-1}f_iX^i$ and $g=\sum_{i=0}^{n-1}g_iX^i$ be two
size-$n$ polynomials. Their product
$h=fg$  is a polynomial of size $2n-1$,
what we call a \emph{balanced full product}.
More generally, if $f$ has size $m$ and $g$ has size $n$, their product
has size $m+n-1$. We call this case the \emph{unbalanced full product} of $f$ and $g$.

We now define precisely the short product, middle product, and
half-additive full product.

\begin{definition}
Let $f$ and $g$ be two size-$n$ polynomials. Their \emph{low short
product} is the size-$n$ polynomial defined as
\[\SPlo(f, g) = (f\cdot g) \bmod X^n\]
and their \emph{high short product} is the size-$(n-1)$ polynomial
defined as
\[\SPhi(f,g) = (f\cdot g) \quo X^n.\]
\end{definition}

The low short product is actually the meaningful notion of product for
truncated power series. Note also that the definition of the high short
product that we use implies that the result does not depend on all the
coefficients of $f$ and $g$. The rationale for this choice is to have
the identity $fg = \SPlo(f,g) + X^n \SPhi(f,g)$.

\begin{definition}\label{def:mp}
Let $f$ and $g$ be two polynomials sizes $n+m-1$ and $n$, respectively.
Their \emph{middle product} is the size-$m$ made of the central coefficients of the product $fg$, that is
\[\MP(f,g) = \left((f\cdot g)\quo X^{n-1}\right)\bmod X^m.\]
\end{definition}

If $f = \sum_{i<n+m-1} f_i X^i$ and $g = \sum_{j < n} g_j X^j$, then
\[\MP(f,g) = \sum_{n-1\le i+j\le n+m-1} f_ig_j X^{i+j-n+1}.\]
The middle product, most commonly in the special case $n=m$, arises
naturally in several algorithms manipulating polynomials or power series
which are based on Newton's iteration, such as division or square root~\cite{Hanrot2004}.

Further, the middle product is obtained by \emph{Tellegen's transposition
principle} from the full product algorithm~\cite{Bostan:2003,Hanrot2004}. This
implies that any full product algorithm yields an algorithm for the middle
product of same time complexity. On the other hand, whether the transposition
can be performed while also preserving the space complexity remains an open
problem~\cite{Kaltofen:2000, Bostan:2003} if one considers the inputs to
be read-only.

\begin{definition}
Let $f$ and $g$ be two polynomials of degree less than $n$, and $h$ be a
polynomial of degree less than $(n-1)$. The \emph{(low-order) half-additive full product} of $f$ and $g$ given $h$ is $\FPlo(f,g,h) = h + fg$.
Similarly, their \emph{high-order half-additive full product} is $\FPhi(f,g,h) = X^nh + fg$.
An \emph{in-place} half-additive full product algorithm is an algorithm computing a half-additive full product where $h$ is initially stored in the output space.
\end{definition}

This variant of the full product which has a partially-initialized
output space will be useful to derive other in-place algorithms.

\subsection{Multiplications as linear maps}\label{sec:linearmaps}

For ease of explanation, we will use the
linear property of polynomial multiplications when an operand is fixed.

Let $f=\sum_{i=0}^{n-1}f_iX^i$ and $g=\sum_{i=0}^{n-1}g_iX^i$ be two
size-$n$ polynomials.
If $f$ is fixed, the product
$h=fg$ can be described as a linear map
from $\K^n$ to $\K^{2n-1}$. The matrix, denoted \mat{\FP(f)}, for this
map is to a Toeplitz matrix built from the coefficients of $f$,
and the product $fg$ corresponds to the following matrix-vector product:
\begin{equation} \label{eq:fullmul}
\small
  \underbrace{
    \begin{pmatrix}
      f_0     &       &     \\
      \vdots  & \ddots &  \\
      f_{n-1} &  &f_0 \\
      &\ddots  &\vdots \\
      &  & f_{n-1} \\
    \end{pmatrix}
  }_{\mat{\FP(f)}}
  \times
  \underbrace{
    \begin{pmatrix}
      g_0 \\
      g_1 \\
      \vdots \\
      g_{n-1} \\
    \end{pmatrix}
  }_{\vec{g}}
  =
  \underbrace{
    \begin{pmatrix}
      h_0 \\
      \\
      h_1 \\
      \\
      \vdots \\
      \\
    h_{2n-1}
  \end{pmatrix}
}_{\vec{h}}
\end{equation}
where $\mat{\FP(f)} \in \K^{(2n-1) \times n}$, $\vec{g}\in\K^n$ and $\vec{h}\in\K^{2n-1}$.

The low and high short products being defined as part of the result of the full
product, their corresponding linear maps are endomorphisms of $\K^n$ and
$\K^{n-1}$ respectively, given by submatrices of \mat{\FP(f)} as follows:
\begin{equation} \label{eq:splow}
\small
\underbrace{
    \begin{pmatrix}
      f_0    & & & \\
      f_1    & \ddots & \\
      \vdots & \ddots &\ddots \\
      f_{n-1} & \hdots & f_1 &f_0 \\
    \end{pmatrix}
  }_{\mat{\SPlo(f)}}
  \quad
  \underbrace{
    \begin{pmatrix}
      f_{n-1} & \hdots & f_2 & f_1 \\
             &  \ddots       &  &   f_2\\
             & & \ddots  & \vdots \\
             &       &  & f_{n-1} \\
    \end{pmatrix}
  }_{\mat{\SPhi(f)}}
\end{equation}

Finally, the middle product corresponds also to a linear map from  $\K^{n}$  to
 $\K^{m}$  when the larger operand is fixed, given by the $m \times n$ Toeplitz
 matrix
\[\underbrace{
    \begin{pmatrix}
    f_{n-1} & f_{n-2} & \dots & f_1 & f_0\\
    f_n     & f_{n-1} &       & f_2 & f_1\\
    \vdots  & \vdots  &       &\vdots& \vdots \\
    f_{n+m-2}&f_{n+m-3}& \dots&f_{m-2}& f_{m-1}
    \end{pmatrix}
  }_{\mat{\MP(f)}}.
\]

\section{Time and space preserving reductions} \label{sec:tisp-preserving}

In this section, we compare the relative difficulties of the full
product, the half-additive full product, the low and high short
products, and the middle product, in the framework of time and space
efficient algorithms. To this end, we define a notion of \emph{time and
space preserving reduction} between problems.

We say that a problem $A$ is \TISP-reducible to a problem $B$ if, given an algorithm for $B$ that has time complexity $t(n)$ and space complexity $s(n)$, one can deduce an algorithm for $A$ that has time complexity $\bigO{t(n)}$ and space complexity $s(n)+\bigO 1$. We write $A\redto B$ is $A$ is \TISP-reducible to $B$ and $A\eqto B$ if both $A\redto B$ and $B\redto A$. Note that the \TISP-reduction is transitive.

The reduction we use can be defined using oracles and is an adaptation of the
notion of \emph{fine-grained reduction}~\cite[Definition 2.1]{Va18} adapted to
time-space fine-grained complexity classes~\cite{LiVaWaWi16}.

\begin{thm}
Half-additive full products and short products are equivalent under \TISP-reductions, that is
\[ \FPhi \eqto \FPlo\eqto \SPhi\eqto\SPlo.\]
Furthermore, if $\SP$ denotes either \SPlo or \SPhi,
\[ \FP\redto\SP \redto \MP. \]
\end{thm}

\begin{proof}
The equivalences $\SPhi\eqto\SPlo$ and $\FPhi\eqto\FPlo$ are proved
below in Lemmas~\ref{lem:shortproducts} and \ref{lem:fullproducts}.
The equivalence $\SP\equiv\FP^+$ (where $\SP$
denotes any of \SPlo and \SPhi, and $\FP^+$ any of \FPlo and \FPhi) is
proved in \cref{sec:equivSPFP}.

The reduction $\FP\redto\SP$ simply amounts to the identity
$\FP(f,g) = \SPlo(f,g) + X^n \SPhi(f,g)$.
The reductions $\SP\redto\MP$ and $\FP\redto\MP$
follow from the following equalities where $0$ denotes the zero
polynomial stored in size $n$:
\begin{align*}
\SPlo(f,g) &= \MP(0 + X^n f, g)\text{,}\\
\SPhi(f,g) &= \MP(f+X^n 0, g)\text{, and}\\
\FP(f, g) &= \MP(0 + X^n f + X^{2n} 0, g).
\end{align*}
Hence, one can compute the full product, the low and high short products of
$f$ and $g$ simply by calling a middle product algorithm on $f$ padded with
zeroes and $g$. In our model of read-only inputs, an actual padding is not
required. It is sufficient to use some kind of \emph{fake padding} where the
data structure storing $f$ is responsible for returning $0$ when needed.
\end{proof}

The relative order of difficulty $\FP\redto\SP\redto\MP$ makes intuitive
sense based on the size of the output compared to the size of the inputs since the
output can be used as work space: The full product maps $2n$ coefficients to
$2n-1$ coefficients, the short products map $2n$ coefficients to $n$
coefficients and the middle product maps $3n$ coefficients to $n$
coefficients. In Section~\ref{sec:out2in}, we shall give a partial converse to
$\SP\redto\MP$: There exists a reduction from $\SP$ to $\MP$ which preserves
space and either maintains the asymptotic complexity or increases it by
a logarithmic factor.

\subsection{Equivalences based on reverse polynomials} \label{sec:equiv_rev}

\begin{definition}
The \emph{size-$n$ reversal} of a polynomial $f$ is $\rev_n(f) = X^{n-1} f(1/X)$.
\end{definition}
We note that any algorithm whose input is a size-$n$ polymial $f$
can be turned into a new algorithm that computes the same function with
input $\rev_n(f)$, simply by replacing a query to any coefficient with
index $i$ with one of index $n-i$, not affecting the number of ring
operations.

Let us now prove that $\SPhi\eqto\SPlo$.
\begin{lem}\label{lem:shortproducts}
Let $f$ and $g$ be two size-$n$ polynomials. Then
\[ \SPhi(f, g) = \rev_{n-1}\left(\SPlo(\rev_{n-1}(f\quo X), \rev_{n-1}(g\quo X))\right).\]
\end{lem}

\begin{proof}
Let $\tilde f = \rev_{n-1}(f\quo X)$ and $\tilde g = \rev_{n-1}(g\quo X)$. Then
\[\SPlo(\tilde f, \tilde g) = \sum_{\substack{0\le i,j < n-1\\i+j < n-1}} f_{n-1-i}g_{n-1-j} X^{i+j},\]
whence
\[\rev_{n-1}\left(\SPlo(\tilde f, \tilde g)\right) = \sum_{\substack{0\le i,j < n-1\\i+j < n-1}} f_{n-1-i} g_{n-1-j} X^{n-2-(i+j)}.\]
One can change the indices of summation using $k = n-1-i$ and $\ell = n-1-j$. Then $n-2-(i+j) = k+\ell-n$ and the indices $i$ and $j$ such that $0\le i+j < n-1$ are mapped to indices $k$ and $\ell$ such that $2n-1 > k+\ell \ge n$. In other words,
\[\rev_{n-1}\left(\SPlo(\tilde f, \tilde g)\right) = \sum_{\substack{0< k,\ell\le n-1\\n\le k+\ell < 2n-1}} f_k g_\ell X^{k+\ell-n} = \SPhi(f,g).\qedhere\]
\end{proof}

Similarly, we can prove that $\FPhi\eqto\FPlo$.
\begin{lem}\label{lem:fullproducts}
Let $f$ and $g$ be two size-$n$ polynomials and $h$ be a size-$(n-1)$ polynomial. Then
\[\FPhi(f,g,h) = \rev_{2n-1}\left(\FPlo(\rev_n(f),\rev_n(g),\rev_{n-1}(h))\right).\]
\end{lem}
\begin{proof}
  Let $f^* = \rev_n(f)$, $g^*=\rev_n(g)$ and $h^*=\rev_{n-1}(h)$.  First note
  that $\rev_{2n-1}(h^*) = X^n h$ by definition.  Since
  $\rev_{2n-1}(f^* g^*) = \rev_n(f^*) \rev_n(g^*)$ we get that
  $\rev_{2n-1}(f^*g^* + h^*) =\rev_n(f^*) \rev_n(g^*) + \rev_{2n-1}(h^*) = fg +
  X^nh= \FPhi(f,g,h)$.
\end{proof}

\subsection{Equivalence between short products and half-additive full products} \label{sec:equivSPFP}

\paragraph{Reduction from $\SP$ to $\FP^{+}$ }
Let $f$ and $g$ be two size-$n$ polynomials and $h$ be a size-$(n-1)$
polynomial. The half-additive full product
$\FPlo(f,g,h)$ equals $fg + h$. Note that $fg = \SPlo(f,g) + X^n
\SPhi(f,g)$. This already proves that the non-additive full product can be
computed using algorithms for low and high short products. For the half-additive
full products, it is sufficient to store an intermediate result in the free
registers of the output space.

Assuming  $\out_{[0.. n-1[}$ holds the value of $h$, the following instructions
reduces the computation of $\FPlo(f,g,h)$ to two short products plus $(n-1)$ additions.
\begin{algorithmic}[1]
\State $\out_{[n-1.. 2n-1[} \gets\algo{SP_{lo}}(f,g)$
\State $\out_{[0.. n-1[} \gets \out_{[0..n-1[} + \out_{[n-1.. 2n-2[}$
\State $\out_{n-1} \gets \out_{2n-1}$
\State $\out_{[n.. 2n-1[}\gets \algo{SP_{hi}}(f,g)$
\end{algorithmic}
\paragraph{Reduction from $\FP^{+}$ to $\SP$}
Let $f$ and $g$ be polynomials of degree less than $n$.
Splitting $f$ and $g$ by half such that
$f = f_0 + X^{\lceil n/2\rceil } f_1$ and $g = g_0 + X^{\lceil n/2\rceil} g_1$,
we have
\[
  \SPlo(f, g)= f_0 g_0 + X^{\lceil n/2\rceil} (f_0g_1+f_1g_0) \bmod
  X^n.
\]
 What is needed is the full product of $f_0$ and $g_0$, and the low short
products of $f_0$ and $g_1$, and $f_1$ and $g_0$. Actually, since $f_0$ is
larger than $g_1$ when $n$ is odd (and $g_0$ larger than $f_1$), one only needs
the short products $\SPlo(f_0^-,g_1)$ and $\SP(f_1, g_0^-)$ where
$f_0^-=f\bmod X^{\lfloor n/2\rfloor}$ and
$g_0^- = g\bmod X^{\lfloor n/2\rfloor}$.

To avoid any recursive call that would imply storing a call stack, we can
actually use full products instead of short products: We first compute $f_0^-g_1 +f_1g_0^-$ using a
full product and a half-additive full product. Then we can forget about the
higher order terms, and add $f_0g_0$ to this sum using a second half-additive
full product. The following instructions summarize this approach:

\begin{algorithmic}[1]
\State $\out_{[0..2\floor{n/2}-1[}\gets \algo{FP}(f_0^-, g_1)$ \Comment half-additivity not needed
\State $\out_{[0..2\floor{n/2}-1[}\gets \algo{FP^+_{lo}}(f_1,g_0^-)$ \Comment
erase higher part of $f_0^- g_1$
\State $\out_{[\ceil{n/2}..n[}\gets \out_{[0..\floor{n/2}[}$ \Comment keep lower
part of $f_0^-g_1 + f_1g_0^-$
\State $\out_{[0..2\ceil{n/2}-1[}\gets \algo{FP_{hi}^+}(f_0,g_0)$
\end{algorithmic}

The correctness is clear. The complexity of the algorithm is the cost of three
full products in degree approximately $n/2$: One non-additive full product in
size $\lfloor n/2\rfloor$ and two half-additive full products in size
$\lfloor n/2\rfloor$ and $\lceil n/2\rceil$, respectively.

\medskip
As direct consequence of Lemmas~\ref{lem:shortproducts}
and~\ref{lem:fullproducts}, one obtains the same reductions to \SPhi and
from \FPlo or \FPhi.

\subsection{From half-additive full product to unbalanced full product} \label{sec:FP2FPu}

The unbalanced full product can be computed using any algorithm for the
(balanced) full product. Nevertheless, the space complexity increases since
intermediate results must be stored. Given an algorithm for the balanced full
product of space complexity $s(n)$, one obtains an algorithm with space
complexity $s(n)+(n-1)$ for the unbalanced full product. In this section, we
prove that if the original full product algorithm is \emph{half-additive}, the
resulting unbalanced full product algorithm has the same space complexity.

Let $f$ be a size-$m$ polynomial and $g$ be a size-$n$ polynomial
with $m > n$. Write $f = \sum_{k=0}^{\ceil{m/n}-1} X^{kn} f_k$, where
each sub-polynomial
$f_0$, \dots, $f_{\ceil{m/n}-1}$ has size at most $n$. The computation of 
$f\cdot g$ reduces to the computations of each $f_k\cdot g$. The following
instructions prove that using half-additivity, the intermediate results 
$f_k\cdot g$ can be computed directly in the output space.

\begin{algorithmic}[1]
    \State $\out_{[\lceil m/n\rceil n..m+n[} \gets \algo{FP}(f_{\lceil
    m/n\rceil}, g)$ \Comment using fake padding
    \For{$k$ from $\ceil{m/n}-1$ down to $0$}
        \State $\out_{[kn.. (k+2)n-1[}\gets \algo{FP^+_{hi}}(f_k,g)$ 
    \EndFor
\end{algorithmic}

Note that at step 1, the polynomial computed may have a larger size that what is needed, due to padding. Yet one can use without difficulty the lower part of the output space to store these additional useless coefficients, that are then erased at step 3.

The time complexity remains $\lceil m/n\rceil\M(n)$ where $\M(n)$ is the complexity of the half-additive full product.

\section{In-place algorithms from out-of-place algorithms} \label{sec:out2in}

In this section, we show how to obtain in-place algorithms from out-of-place algorithms. The theorem below summarizes the main results described in this section.

\begin{thm}\label{thm:ifromo}
\begin{enumerate}
\item Given a full product algorithm with time complexity $\M(n)$ and space complexity $\le cn$, one can build an in-place algorithm for the half-additive full product with time complexity $\le (2c+7)\M(n) + o(\M(n))$.
\item Given a (low or high) short product algorithm with time complexity $\M(n)$ and space complexity $\le cn$, one can build an in-place algorithm for the same problem with time complexity $\le (2c+5)\M(n) + o(\M(n))$.
\item Given a middle product algorithm with time complexity $\M(n)$ and space complexity $\le cn$, one can build an in-place algorithm for the same problem with time complexity $\le \M(n)\log_{\frac{c+1}{c+2}}(n) + \bigO{\M(n)}$ if $\M(n)$ is quasi-linear, and $\bigO{\M(n)}$ otherwise.
\end{enumerate}
\end{thm}

Actually, our reductions work for any space bound $s(n)\le\bigO n$. Smaller space bounds yield better time bounds though we do not have a general expression in terms of $s(n)$. Yet sublinear space bounds still imply an increase of the time complexity by a multiplicative constant for full and short products.

Formally, we give self-reductions for the three problems. That is, we use an out-of-place algorithm for the problem as building block of our in-place version. The general idea is similar in the three cases. In a first step, we use the out-of-place algorithm to compute some part of the output, using the unused output space as temporary work space. Then a recursive call finishes the work. The (constant) amount of space needed in our in-place algorithms correspond the space needed to process the base cases.

Using the language of linear algebra, we aim to apply some specific matrix to a vector. The general construction we use consists in first applying the top or bottom rows of the matrix to the vector using the out-of-place algorithm, and applying the remaining rows using a recursive call (\emph{cf.} Fig.~\ref{fig:decompositions}). In the cases of full and short products, the diamond and triangular shapes of the corresponding matrices imply that the recursive call is made on two smaller inputs: For instance, to apply the first rows of a triangular matrix to a vector, one only needs to apply it to the first entries of the vector. For the middle product, the square shape imply that one input remains of the same size in the recursive call. This difference explains the difference in the time complexities in Theorem~\ref{thm:ifromo}.

\begin{figure*}[t]
\svg[\textwidth]{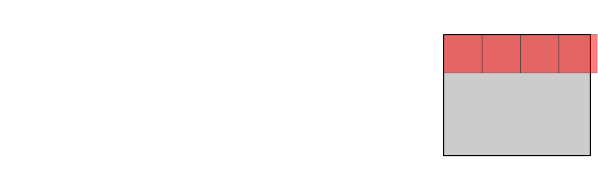}
\caption{Tilings of the matrices $\mat{\FP(f)}$ (left), $\mat{\SPlo(f)}$ (center) and $\mat{\MP(f)}$ (right).}
\label{fig:decompositions}
\end{figure*}

\subsection{In-place full product algorithm} \label{sec:oFP2iFP}

Our aim is to build an in-place (low-order) half-additive full product algorithm \iFPhi
based on an out-of-place full product algorithm \oFP that has space complexity $cn$. 
That is, we are given two polynomials $f$ and $g$ of degree $<n$
in the input space and a polynomial $h$ of degree $<n-1$ in the $(n-1)$
low-order registers of the output space $\out$ and we aim to compute $fg + h$ in
$\out$. The algorithm is based on the tiling of the matrix $\mat{\FP(f)}$
given in Fig.~\ref{fig:decompositions} (left). 

For some $k < n$
to be fixed later, let $f= \hat{f}X^k+f_0$ and $g= \hat{g}X^k+g_0$
where $\deg f_0, \deg g_0 <k$. Then we have
\begin{equation}\label{eq:recursiveIFP}
  h+fg = h+f_0g+\hat{f}g_0X^k+ \hat{f}\hat{g}X^{2k}.
\end{equation}
Recall that the output
\out~has size $2n-1$ with its $n-1$ lowest registers containing $h$.
Then equation \eqref{eq:recursiveIFP} can be evaluated with the following three
steps:
\begin{algorithmic}[1]
\State  $ \out_{[0..n+k-1[} \gets h + f_0g $
\State  $ \out_{[k..n+k-1[} \gets \out_{[k..n+k-1[} + \hat{f}g_0 $
\State $\out_{[2k..2n[}  \gets \out_{[2k..2n[} + \hat{f}\hat{g}$
\end{algorithmic}
The first two steps corresponds exactly to two \emph{additive} unbalanced full 
products, that is unbalanced full products that must be added to some already
filled output space. One can describe an algorithm \oFPunbalhA for this task, 
based on a (standard) full product algorithm \oFP: If $f$ has degree $< k$ and
$g$ has degree $< n$, $n > k$, we write $g = \sum_{i=0}^{\ceil{n/k}-1} g_i 
X^{ki}$ with $\deg(g_i) < k$. Then $fg = \sum_i fg_i$: The algorithm computes
the $\ceil{n/k}$ products $fg_i$ in $2k-1$ extra registers and adds them
to the output. If \oFP has time complexity $\M(n)$ and space complexity $cn$,
the time complexity of \oFPunbalhA is $\ceil{n/k}(\M(k) + 2k-1)$ and its space
complexity $(c+2)k-1$. 

The last step computes $h+fg$ and
corresponds to a half-additive full product on inputs of degree $<n-k$, since
only the $n-k-1$ first registers of $\out_{[2k..2n[}$ are filled: Indeed,
$ \deg (h + f_0g +\hat{f}g_0X^k)< n+k-1$. This last step is thus a recursive
call.

In order to make this algorithm run in place, $k$ must be chosen so that the
extra memory needed in the two calls to \oFPunbalhA fits exactly in the unused
part of \out. This is the case when
\[
  (c+2)k-1 \leq 2n-1 - (n+k-1)
\]
which gives $k\leq \frac{n+1}{c+3}$. The resulting algorithm is formally depicted below.

\begin{algorithm}
\caption{\reduction{\iFPhi}{\oFP}}
\begin{algorithmic}[1]
\Input $f$ and $g$ of degree $< n$ in the input space, $h$ of degree $<n-1$ in the output space $\out$
\Output $\out$ contains $fg+h$
\BaseAlg A full product algorithm \oFP with space complexity $\leq c n$

    \If{$n < c+2$}
        \State $\out\gets \out + fg$ \Comment using a naive algorithm
    \Else
    \State $k\gets\floor{(n+1)/(c+3)}$
    \State $\out_{[0..n+k-1[}\gets \oFPunbalhA(h,f_0,g)$ \Comment work space: $\out_{[n+k-1..2n[}$ %
    \State $\out_{[k..n+k-1[}\gets \oFPunbalhA(h+f_0g,f,g_0)$  \Comment same work space 
    \State $\out_{[2k.. 2n[}\gets\reduction{\iFPhi}{\oFP}(f\quo X^k, g\quo X^k)$
    \EndIf
\end{algorithmic}
\end{algorithm}

\paragraph{Complexity analysis} The algorithm uses two calls to \oFPunbalhA
with inputs of sizes $(k,n)$ and $(n-k,k)$ respectively. The total complexity
amounts to
$\ceil{n/k} \M(k) + (\ceil{n/k]}-1)\M(k) + 2(\ceil{n/k}-1)(2k-1)$ plus a recursive call in size $n-k$. Let $T(n)$ be the complexity of
\iFPhi, we thus have
\[
  T(n) = T(n-k) + (2\lceil n/k \rceil -1) \left[\M(k)+(2k-1)\right] .
\]
Note that $k$ depends upon $n$, this implies that the analysis must be
done without $k$.  Since
$k=\lfloor (n+1)/(c+3)\rfloor$, $ \lceil n/k\rceil \leq c+4$
for $n \geq (c+2)(c+4)$.
Therefore,
\[  T(n) \leq T\left(\frac{c+2}{c+3}(n+1)\right) + (2c+7)\left[
    \M\left(\frac{n+1}{c+3}\right) + 2\frac{n}{c+3}- \frac{c+1}{c+3}\right].
\]
Using Corollary~\ref{cor:recurrence}, we conclude that $T(n) \le (2c+7)\M(n) + o(\M(n))$.

\subsection{In-place short product algorithm} \label{sec:oSP2iSP}

Our goal is to describe an in-place (low) short product algorithm based on an out-of-place one, based on the tiling of $\mat{\SPlo(f)}$ depicted on Fig.~\ref{fig:decompositions} (center). Let $f=\sum_{i=0}^{n-1} f_i X^i$ and $g = \sum_{i=0}^{n-1} g_i X^i$, and let $h = \sum_{i=0}^{n-1} h_i X^i = \SPlo(f,g)$. The idea is to fix some $k < n$ and to have two phases. The first phase corresponds to the bottom $k$ rows of $\mat{\SPlo(f)}$ and computes $h_{n-k}$ to $h_{n-1}$ using the out-of-place algorithm on smaller polynomials. The second phase corresponds to the top $(n-k)$ rows and is a recursive call to compute $h_0$ to $h_{n-k-1}$: Indeed, $h\bmod X^{n-k} = \SPlo(f\bmod X^{n-k}, g\bmod X^{n-k})$.

For the second phase, we remark that the bottom $k$ rows can be tiled by $\ceil{n/k}$ lower triangular matrices (denoted $L_0$, \dots, $L_{\ceil{n/k}-1}$ from the right to the left), and $\ceil{n/k}-1$ upper triangular matrices (denoted $U_0$, \dots, $U_{\ceil{n/k}-2}$).
One can identify the matrices $L_i$ and $U_i$ as matrices of some low and high short products. More precisely, the coefficients that appear in the lower triangular matrix $L_i$ are the coefficients of degree $ki$ to $k(i+1)-1$ of $f$. Thus, $L_i = \mat{\SPlo(f_{ki, k(i+1)})}$ where $f_{ki,k(i+1)} = \sum_{j=ki}^{k(i+1)-1} f_j X^{j-ki}$. Similarly, 
$U_i = \mat{\SPhi(f_{ki,k(i+1)})}$. The matrices $L_{\ceil{n/k}-1}$ and $U_{\ceil{n/k}-2}$ must be padded if $k$ does not divide $n$. Altogether, this proves that this part of the computation reduces to $\ceil{n/k}$ low short products and $\ceil{n/k}-1$ high short products, in size $k$.

In order for this algorithm to actually be in place, $k$ must be small enough. If the out-of-place short product algorithm uses $ck$ extra space, since we also need $k$ free registers to store the intermediate results, $k$ must satisfy $n-k\ge (c+1)k$, that is $k\le n/(c+2)$.

\begin{algorithm}
\caption{\reduction{\iSPlo}{\oSP}} 
\begin{algorithmic}[1]
\Input $f$ and $g$ of degree $<n$
\Output $\out$ contains $\SPlo(f,g)$
\BaseAlg Two short product algorithms \algo{oSP_{lo}} and \algo{oSP_{hi}} with space complexity $\le cn$
    \If{$n < c+2$}
        \State $\out\gets \SPlo(f,g)$ \Comment using a naive algorithm
    \Else
        \State $k\gets \floor{n/(c+2)}$
        \For{$i=0$ to $\ceil{n/k}-1$}
        \Comment work space: $\out_{[0..n-k[}$
        \State $\out_{[n-k..n[} +\!=\oSPlo(f_{ki,k(i+1)}, g_{n-k(i+1), n-ki)})$
        \EndFor
        \For{$i=0$ to $\ceil{n/k}-2$} \Comment same work space
        \State $\out_{[n-k..n[} +\!= \oSPhi(f_{ki,k(i+1)}, g_{n-k(i+2), n-k(i+1)})$
        \EndFor
        \State $\out_{[0.. n-k[}\gets \reduction{\iSPlo}{\oSP}(f \bmod X^{n-k}, g\bmod X^{n-k})$ 
    \EndIf
\end{algorithmic}
\end{algorithm}

\paragraph{Complexity analysis} The algorithm performs $\ceil{n/k}$ low short products and $\ceil{n/k}-1$ high short products plus one recursive call in size $n-k$. Let $\M(k)$ be the complexity of a low short product algorithm. Then the high short product can be computed in time $\M(k-1)$. Let $T(n)$ be the complexity of the recursive algorithm. Then $T(n) = \ceil{n/k} \M(k) + (\ceil{n/k}-1)\M(k-1) + 2(\ceil{n/k}-1)k + T(n-k)$ (the linear time is for the additions). Since $k=\floor{n/(c+2)}$, $\ceil{n/k}\le c+3$ for $n\ge (c+3)(c+2)$ and $n-k\le \frac{c+1}{c+2} n + 1$. Thus,
\[T(n)\le (c+3)\M\left(\frac{n}{c+2}\right) + (c+2)\M\left(\frac{n}{c+2}-1\right) + 2n+T\left(\frac{c+1}{c+2}n+1\right).\]
Using Corollary~\ref{cor:recurrence}, this equation yields $T(n)\le (2c+5)\M(n) + o(\M(n))$.

\subsection{In-place middle product algorithm} \label{sec:oMP2iMP} 

To build an in-place middle product algorithm, we assume that we have an algorithm for the middle product that uses $cn$ extra space to compute the middle product in size $(n,m)$ (that is with inputs of degree $<n+m-1$ and $<n$, respectively).

The in-place algorithm is again based on the tiling given in Fig.~\ref{fig:decompositions} (right): The top $k$ rows correspond to the matrix \mat{\MP(f\bmod X^k)} and the bottom $m-k$ rows to the matrix \mat{\MP(f\quo X^k)}. The algorithm consists in computing $\mat{\MP(f\bmod X^k)}\vec g$ using the out-of-place algorithm and then $\mat{\MP(f\quo X^k)}\vec g$ using a recursive call.

To make this algorithm work in place, the value of $k$ has to be adjusted so that the work space is large enough. The result of a middle product in size $k$ has degree $<k$ and needs $ck$ extra work space by hypothesis. Therefore, if $m-k \ge (c+1)k$, that is $k\le m/(c+2)$, the computation can be performed in place.

\begin{algorithm}
\caption{\reduction{\iMP}{\oMP}}
\begin{algorithmic}[1]
\Input $f$ and $g$ of degree $< n+m-1$ and $<n$ respectively
\Output $\out$ contains $\MP(f,g)$
\BaseAlg An out-of-place middle product algorithm \oMP with space complexity $\le cn$
    \If{$m < c+2$} 
        \State $\out\gets\oMP(f,g)$ \Comment using a naive algorithm
    \Else
        \State $k\gets\floor{m/(c+2)}$
        \State $\out_{[0..k[}\gets \oMP(f\bmod X^{n+k}, g)$ \Comment work space: $\out_{[k..m[}$
        \State $\out_{[k.. m[}\gets \reduction{\iMP}{\oMP}(f\quo X^k, g)$  \Comment recursive call
    \EndIf
\end{algorithmic}
\end{algorithm}

\paragraph{Complexity analysis} Let $\M(k)$ be the cost of an out-of-place balanced middle product algorithm. The cost of an unbalanced middle product is thus $\ceil{n/k}\M(k)$ for $k < n$. The in-place algorithm computes first a middle product using an out-of-place algorithm and then makes a recursive call on the remaining part. Note that $n$ does not change during the algorithm and can be viewed as a large constant, while $m$ is the parameter that varies. Then the cost of the algorithm verifies $T(m)\le \ceil{n/k} \M(k) + T(m-k)$. Since $k = \floor{m/(c+2)}$, $\ceil{n/k} < n(c+2)/(m-c-2)+1$ and $m-k \le (c+1)m/(c+2)+1$. Furthermore, $\M(k)\le m/n(c+2) \M(n)$, thus $\ceil{n/k}\M(k)\le (m/(m-c-2) + m/n(c+2))\M(n)$. That is,
\[T(m)\le \left(\frac{m}{n(c+2)} + \frac{c+2}{m-c-2} + 1\right)\M(n) + T\left(\frac{c+1}{c+2} m +1\right).\]
Corollary~\ref{cor:rec2} implies $T(n)\le \M(n)\log_{\frac{c+2}{c+1}}(n) + O(\M(n))$ for $m = n$.

\paragraph{Improvement for non quasi-linear algorithms} The extra logarithmic factor only occurs when $\M(n) =n^{1+o(1)}$. Suppose to the contrary that $\M(n)\le\lambda n^\gamma$ for some $\gamma > 1$. The recurrence now reads $T(m)\le \left(\frac{n(c+2)}{m-c-2}+1\right) \lambda\left(\frac{m}{c+2}\right)^\gamma + T(\frac{c+1}{c+2}m + 1)$. We claim that there exist constants $\mu$ and $\nu$ such that $T(m)\le \mu m^{\gamma-1}n + \nu m^\gamma + o(m^{\gamma-1}n+m^\gamma)$ and prove it by induction. Using the recurrence relation and the induction hypothesis,
\begin{multline*}
    T(m)\le \frac{\lambda nm^{\gamma-1}}{(c+2)^{\gamma-1}} + \frac{\lambda m^\gamma}{(c+2)^\gamma} 
    + \mu\left(\frac{c+1}{c+2}\right)^{\gamma-1} m^{\gamma-1}n \\+ \nu\left(\frac{c+1}{c+2}\right)^\gamma m^\gamma
    + o(m^{\gamma-1}n + m^\gamma).
\end{multline*}
The result follows as soon as $(\lambda+\mu(c+1)^{\gamma-1})/(c+2)^{\gamma-1}\le\mu$ and $(\lambda+\nu(c+1)^{\gamma})/(c+2)^{\gamma}\le\nu$. We can thus fix
\[\mu = \frac{\lambda}{(c+2)^{\gamma-1}-(c+1)^{\gamma-1}}\text{ and }\nu=\frac{\lambda}{(c+2)^\gamma-(c+1)^\gamma}.\]
Finally, taking $m = n$, we conclude that $T(n) \le (\mu+\nu)\lambda n^\gamma +O(n^{\gamma-1})$.

\paragraph{Reduction from short products to middle product}
The middle product of $f$ and $g$ can be computed as the sum of the low short product of $f\quo X^n$ with $g$ and the high short product of $f\bmod X^n$ with $g$. Yet this reduction does not preserve the space complexity since one needs to store the results of the two short products in two zones of size $n$ before summing them. Actually, the reduction given above from \oMP to \iMP can easily be adapted to a reduction from \SP to \MP that is space-preserving. Yet, the complexity also worsens with a logarithmic factor. Thus, we cannot conclude that $\MP\redto\SP$. 

\subsection{Resolution of recurrences}

\begin{lem}\label{lem:recsum}
Let $T(n)$ be a function satisfying $T(n) \le f(n) + T(\floor{\alpha n + \beta})$ for some $\alpha < 1$. Then
\[T(n) \le T(\floor{n_K}) + \sum_{i=0}^{K-1} f(n_i)\]
where $n_i = \alpha^i n + \beta\frac{1-\alpha^{i+1}}{1-\alpha}$ and $K \le \log_{1/\alpha}(n)$.
\end{lem}

\begin{proof}
Let $T(x) = T(\floor x)$ for non integral $x$. By definition of $n_i$, $n = n_0$ and $T(n_i)\le f(n_i) + T(n_{i+1})$. Then by recurrence, $T(n) \le T(n_{i+1}) + \sum_{j=0}^i f(n_i)$. 
\end{proof}

\begin{lem}\label{lem:sumni}
Let $n_i = \alpha^i n + \beta\frac{1-\alpha^{i+1}}{1-\alpha}$. Then
\[\sum_{i=0}^{K-1} n_i 
    \le \frac{n+\beta K}{1-\alpha}.\]
\end{lem}

\begin{proof}
Since $\sum_{i=0}^{K-1} \alpha^i = (1-\alpha^K)/(1-\alpha)$ and $1-\alpha > 0$, $\sum_i \alpha^in\le n/(1-\alpha)$. Then, $\sum_i (1-\alpha^{i+1})/(1-\alpha) = K/(1-\alpha) + (\alpha^{K+1}-\alpha)/(1-\alpha)^2\le K/(1-\alpha)$ since $\alpha^{K+1} < \alpha$.
\end{proof}

\begin{lem}\label{lem:suminvni}
Let $n_i = \alpha^i n + \beta\frac{1-\alpha^{i+1}}{1-\alpha}$. Then
\[\sum_{i=0}^{K-1} \frac{1}{n_i-\beta/(1-\alpha)} = \frac{\alpha(\alpha^{-K}-1)}{(1-\alpha)n-\alpha\beta}.\]
\end{lem}

\begin{proof}
Since $n_i = \alpha^i(n - \beta\alpha/(1-\alpha)) + \beta/(1-\alpha)$, $n_i - \beta/(1-\alpha)$ is a multiple of $\alpha^i$. Thus,
\[\sum_{i=0}^{K-1} \frac{1}{n_i-\beta/(1-\alpha)} = \frac{1}{n-\beta\alpha/(1-\alpha)} \sum_{i=0}^{K-1} \alpha^{-i}.\]
Then, $\sum_i \alpha^{-i} = (1-\alpha^{-K})/(1-1/\alpha) = \alpha(\alpha^{-K}-1)/(1-\alpha)$, and $\sum_i 1/(n_i-\beta/(1-\alpha)) = \alpha(\alpha^{-K}-1)/((1-\alpha)n-\alpha\beta)$.
\end{proof}

\begin{lem}\label{lem:sumMni}
If $\M(n)/n$ is non-decreasing, and $n_i = \alpha^i n + \beta(1-\alpha^{i+1})/(1-\alpha)$ for some $\alpha < 1$, then
\[\sum_{i=0}^{K-1} \M(\lambda n_i+\mu) 
    = \frac{\lambda}{1-\alpha} \M(n) + o(\M(n))\]
for $K \le\log_{1/\alpha}(n)$ and any $\lambda$ and $\mu$ such that $\lambda n_i+\mu \le n$ for all $n_i$.
\end{lem}

\begin{proof}
Since $\M(n)/n$ is non-decreasing, $M(\lambda n_i+\mu) \le \frac{\lambda n_i+\mu}{n}\M(n)$. Therefore, $\sum_i \M(\lambda n_i+\mu) \le \M(n)/n \sum_i \lambda n_i + \mu$. By Lemma~\ref{lem:sumni}, $\sum_i\M(\lambda n_i+\mu)\le \lambda \M(n)/(1-\alpha) + \lambda \beta K \M(n)/n(1-\alpha) + \mu K\M(n)/n$. Since $K = O(\log n)$, $K\M(n)/n = o(\M(n))$.
\end{proof}

\begin{cor}\label{cor:recurrence}
Let $T(n) \le \sum_k a_k \M(\lambda_k n+\mu_k) + bn+c + T(\alpha n + \beta)$ with $\alpha < 1$ and $\lambda_k n + \mu_k < n$ for all $k$. Then
\[T(n) \le \sum_k \frac{a_k\lambda_k}{1-\alpha} \M(n) + \frac{bn}{1-\alpha} + o(\M(n)).\]
The linear term is negligible but if $\M(n) = \bigO{n}$.
\end{cor}

\begin{proof}
By Lemma~\ref{lem:recsum}, $T(n)\le T(n_K) + \sum_i f(n_i)$ with $n_i$ defined as in the lemma and $f(n) = \sum_k a_k \M(\lambda_k n+\mu_k)+bn+c$. Then
\begin{align*}
\sum_{i=0}^{K-1} f(n_i) &= \sum_k a_k \sum_{i=0}^{K-1} \M(\lambda_k n_i+\mu_k) + b\sum_{i=0}^{K-1} n_i + Kc\\
& \le \sum_k  a_k \left(\frac{\lambda_k}{1-\alpha} \M(n) + o(\M(n))\right) + b\frac{n+\beta K}{1-\alpha} + Kc\\
& = \sum_k \frac{a_k\lambda_k}{1-\alpha} \M(n) + \frac{bn}{1-\alpha} + o(\M(n))
\end{align*}
since $K = o(\M(n))$ and the sum over $k$ is of fixed size.
\end{proof}

\begin{cor}\label{cor:rec2}
Let $T(m)\le  (\lambda m/n + \mu/(m-\frac{1}{1-\alpha}) + 1)\M(n) + T(\alpha m+ 1)$ with $\alpha < 1$ and $m\le n$. Then for $m = n$,
\[T(n) \le \M(n)\log_{1/\alpha}(n) + \frac{\lambda + \mu\alpha}{1-\alpha} \M(n) + o(\M(n)).\]
\end{cor}

\begin{proof}
By Lemma~\ref{lem:recsum},
\[T(m)\le T(m_K) + \M(n) \sum_i \left(\frac{\lambda m_i}{n} + \frac{\mu}{m_i-1/(1-\alpha)} + 1\right)\]
where $m_i = \alpha^i m + (1-\alpha^{i+1})/(1-\alpha)$. By Lemma~\ref{lem:sumni}, $\sum_i m_i \le (m+K)/(1-\alpha)$ and by Lemma~\ref{lem:suminvni}, $\sum_i 1/(m_i-\frac{1}{1-\alpha}) \le \alpha^{-K+1}/((1-\alpha)m-\alpha)$. Altogether,
\[T(m)\le T(m_K) + K\M(n) + \frac{\lambda(m+K)}{n(1-\alpha)} \M(n) + \frac{\mu\alpha}{1-\alpha}\cdot \frac{(1/\alpha)^K}{m-\alpha/(1-\alpha)}\M(n).\]
If we plug $K = \log_{1/\alpha}(m)$ and fix $m = n$, we get
\[T(n)\le T(n_K) + \M(n)\log_{1/\alpha} n + \frac{\lambda + \mu\alpha}{1-\alpha}\M(n) + o(M(n)).\]
\end{proof}

\section{Perspectives}

We have presented algorithms for polynomial multiplication problems
which are efficient in terms of both time and space.
Our results show that any
algorithm for the full and short products of polynomials can be turned into
another algorithm
with the same asymptotic time complexity while using only $\bigO{1}$
extra space. We obtain similar results for the middle product but only
proved it for algorithms that do not have a quasi-linear time complexity. In the
latter case, an increase of the time complexity by a logarithmic factor occurs.
We provided analysis of our reductions that make their constants explicit. In
particular, their values ensure that our reductions are practicable.

In a future work, we plan to address some remaining issues. By examining the
constants in the already known algorithms, we can choose the algorithms to use
as starting points of our reductions to optimize the complexity. For instance
three variants of Karatsuba's algorithm with different time and space 
complexities are known~\cite{roche:2009,Thome:2002,karatsuba}. Furthermore,
it seems possible to improve on the complexity of low-space versions of
Karatsuba's and Toom-Cook's algorithm, yielding faster in-place algorithms
through our reductions. Another promising approach is to slightly relax the
model of computation and work in model in which one can write on the input
space as long as the original inputs are restored by the end of the
computation. Preliminary results for Karatsuba's algorithm suggest that this
could also yield a lower constant in the time complexity.

Finally, we have stated to explore the design of in-place
algorithms for a broader range of problems of polynomials, such as division
or evaluation/interpolation. The use of in-place middle and short products
becomes crucial since one needs to avoid any increase in the size of the
intermediate results.

\section*{Acknowledgements}
This work was begun while the last author was graciously
  hosted by the LIRMM at the Universit\'e Montpellier.

  This work was supported in part by the
  \grantsponsor{nsf}{National Science Foundation}{https://nsf.gov/}
  under grants
  \grantnum[https://www.nsf.gov/awardsearch/showAward?AWD_ID=1319994]{nsf}{1319994}
  and
  \grantnum[https://www.nsf.gov/awardsearch/showAward?AWD_ID=1618269]{nsf}{1618269}.
\balance

\bibliographystyle{abbrv}
\input{issac-report.bbl}

\end{document}

%% file: issac-report.bbl
\newcommand{\Gathen}{\relax}\newcommand{\Hoeven}{\relax}

%% file: issac-report.bbl
\begin{thebibliography}{10}

\bibitem{Abr86}
K.~Abrahamson.
\newblock Time-space tradeoffs for branching programs contrasted with those for
  straight-line programs.
\newblock In {\em 27th {Annual} {Symposium} on {Foundations} of {Computer}
  {Science} (sfcs 1986)}, pages 402--409, 1986.

\bibitem{ArBa09}
S.~Arora and B.~Barak.
\newblock {\em Computational {Complexity}: {A} {Modern} {Approach}}.
\newblock Cambridge University Press, 1st edition, 2009.

\bibitem{BCG+17}
A.~Bostan, F.~Chyzak, M.~Giusti, R.~Lebreton, G.~Lecerf, B.~Salvy, and
  E.~Schost.
\newblock {\em Algorithmes Efficaces en Calcul Formel}.
\newblock 1.0 edition, Aug. 2017.

\bibitem{Bostan:2003}
A.~Bostan, G.~Lecerf, and E.~Schost.
\newblock Tellegen's principle into practice.
\newblock In {\em Proceedings of the 2003 International Symposium on Symbolic
  and Algebraic Computation}, ISSAC '03, pages 37--44, New York, NY, USA, 2003.
  ACM.

\bibitem{Brent:2010:MCA}
R.~Brent and P.~Zimmermann.
\newblock {\em Modern Computer Arithmetic}.
\newblock Cambridge University Press, New York, NY, USA, 2010.

\bibitem{CK91}
D.~G. Cantor and E.~Kaltofen.
\newblock On fast multiplication of polynomials over arbitrary algebras.
\newblock {\em Acta Informatica}, 28:693--701, 1991.

\bibitem{Che16}
Y.~Cheng.
\newblock Space-efficient karatsuba multiplication for multi-precision
  integers.
\newblock {\em CoRR}, abs/1605.06760, 2016.

\bibitem{Cook66}
S.~A. Cook.
\newblock {\em On the minimum computation time of functions}.
\newblock PhD thesis, Harvard University, May 1966.

\bibitem{Gathen:2013}
J.~v.~z. Gathen and J.~Gerhard.
\newblock {\em Modern Computer Algebra (third edition)}.
\newblock Cambridge University Press, 2013.

\bibitem{HQZ00}
G.~Hanrot, M.~Quercia, and P.~Zimmermann.
\newblock {Speeding up the Division and Square Root of Power Series}.
\newblock Technical Report RR-3973, {INRIA}, 2000.

\bibitem{Hanrot2004}
G.~Hanrot, M.~Quercia, and P.~Zimmermann.
\newblock The middle product algorithm i.
\newblock {\em Applicable Algebra in Engineering, Communication and Computing},
  14(6):415--438, Mar 2004.

\bibitem{HanrotZimmerman:2004}
G.~Hanrot and P.~Zimmermann.
\newblock A long note on {Mulders}' short product.
\newblock {\em Journal of Symbolic Computation}, 37(3):391--401, 2004.

\bibitem{HHL17}
D.~Harvey, J.~\Hoeven{van der Hoeven}, and G.~Lecerf.
\newblock Faster polynomial multiplication over finite fields.
\newblock {\em J. ACM}, 63(6):52:1--52:23, Jan. 2017.

\bibitem{hr10}
D.~Harvey and D.~S. Roche.
\newblock An in-place truncated {F}ourier transform and applications to
  polynomial multiplication.
\newblock In {\em ISSAC '10: Proceedings of the 2010 International Symposium on
  Symbolic and Algebraic Computation}, pages 325--329, New York, NY, USA, 2010.
  ACM.

\bibitem{Kaltofen:2000}
E.~Kaltofen.
\newblock Challenges of symbolic computation: my favorite open problems.
\newblock {\em Journal of Symbolic Computation}, 29(6):891--919, 2000.

\bibitem{karatsuba}
A.~Karatsuba and Y.~Ofman.
\newblock {Multiplication of Multidigit Numbers on Automata}.
\newblock {\em Soviet Physics-Doklady}, 7:595--596, 1963.

\bibitem{LiVaWaWi16}
A.~Lincoln, V.~Vassilevska~Williams, J.~R. Wang, and R.~R. Williams.
\newblock Deterministic {Time}-{Space} {Trade}-{Offs} for k-{SUM}.
\newblock In I.~Chatzigiannakis, M.~Mitzenmacher, Y.~Rabani, and D.~Sangiorgi,
  editors, {\em 43rd {International} {Colloquium} on {Automata}, {Languages},
  and {Programming} ({ICALP} 2016)}, volume~55 of {\em Leibniz {International}
  {Proceedings} in {Informatics} ({LIPIcs})}, pages 58:1--58:14, Dagstuhl,
  Germany, 2016. Schloss Dagstuhl–Leibniz-Zentrum fuer Informatik.

\bibitem{Mulders:2000}
T.~Mulders.
\newblock On {Short} {Multiplications} and {Divisions}.
\newblock {\em Applicable Algebra in Engineering, Communication and Computing},
  11(1):69--88, 2000.

\bibitem{roche:2009}
D.~S. Roche.
\newblock Space- and time-efficient polynomial multiplication.
\newblock In {\em Proceedings of the 2009 International Symposium on Symbolic
  and Algebraic Computation}, ISSAC '09, pages 295--302. ACM, 2009.

\bibitem{SS79}
J.~Savage and S.~Swamy.
\newblock Space-time tradeoffs for oblivious integer multiplication.
\newblock In H.~Maurer, editor, {\em Automata, Languages and Programming},
  volume~71 of {\em Lecture Notes in Computer Science}, pages 498--504.
  Springer Berlin / Heidelberg, 1979.

\bibitem{SS71}
A.~Sch{\"o}nhage and V.~Strassen.
\newblock {Schnelle Multiplikation gro{\ss}er Zahlen}.
\newblock {\em Computing}, 7:281--292, 1971.

\bibitem{SF12}
C.~Su and H.~Fan.
\newblock Impact of {I}ntel's new instruction sets on software implementation
  of {GF(2)[x]} multiplication.
\newblock {\em Information Processing Letters}, 112(12):497--502, 2012.

\bibitem{Thome:2002}
E.~Thomé.
\newblock Karatsuba multiplication with temporary space of size $\leq$ n.
\newblock online, 2002.

\bibitem{Va18}
V.~Vassilevska~Williams.
\newblock On some fine-grained questions in algorithms and complexity.
\newblock In {\em Proceedings {ICM}}, 2018.

\end{thebibliography}
